\documentclass{svproc}

\usepackage{url}

\usepackage[utf8]{inputenc}
\usepackage{amsmath}
\usepackage{amsfonts}
\usepackage{amssymb}
\usepackage{version}
\usepackage{graphicx}
\usepackage{subcaption}
\usepackage{xcolor}

\usepackage{booktabs} 
\usepackage{xspace}
\usepackage{dsfont}
\usepackage{siunitx}
\usepackage{subcaption}
\usepackage[symbol]{footmisc}
\renewcommand{\thefootnote}{\fnsymbol{footnote}}

\usepackage{algorithm}
\usepackage[noend]{algpseudocode}
\usepackage{color}

\newtheorem{condition}{Condition}

\newcommand{\fredvu}[1]{}

\newcommand{\thibvu}[1]{}
\def\nitbf{\noindent\textbf}

\def\ms{\medskip}

\graphicspath{{figures/}}

\newcommand{\DD}{DD\xspace}
\newcommand{\DDs}{DDs\xspace}
\newcommand{\DDtot}{degree distribution\xspace}
\newcommand{\DDstot}{degree distributions\xspace}

\renewcommand{\Pr}{\mathbb{P}}

\newcommand{\E}{{\mathbb{E}}}

\includeversion{longversion}

\excludeversion{ComplexNet}

\begin{document}

\mainmatter              

\title{A Random Growth Model with any Real or Theoretical Degree Distribution\thanks{This work has been supported by the French government through the UCA JEDI (ANR-15-IDEX-01) and EUR DS4H (ANR-17-EURE-004) Investments in the Future projects, by the SNIF project, and by Inria associated team EfDyNet.}}

\author{
Thibaud Trolliet\inst{1} \and
Frédéric Giroire\inst{2} \and
Stéphane Pérennes\inst{2}
}

\authorrunning{T.Trolliet, F.Giroire, S.Pérennes}
\titlerunning{A Random Growth Model with any Degree Distribution}
%
\institute{INRIA Sophia-Antipolis, France \and
Université Côte d'Azur/CNRS, France }
\maketitle              
\begin{abstract}
The degree distributions of complex networks are usually considered to be power law. However, it is not the case for a large number of them.
We thus propose a new model able to build random growing networks with (almost) any wanted degree distribution. The degree distribution can either be theoretical or extracted from a real-world network. The main idea is to invert the recurrence equation commonly used to compute the degree distribution in order to find a convenient attachment function for node connections - commonly chosen as linear. We compute this attachment function for some classical distributions, as the power-law, broken power-law, geometric and Poisson distributions. We also use the model on an undirected version of the Twitter network, for which the degree distribution has an unusual shape.
We finally show that the divergence of chosen attachment functions is heavily links to the heavy-tailed property of the obtained degree distributions.
\keywords{Complex Networks, Random Growth Model, Preferential Attachment, Degree Distribution, Twitter, Heavy-Tailed Distributions}
\end{abstract}


\section{Introduction}

\begin{figure}[t]
	\begin{subfigure}{.24\textwidth}
		\centering
		\includegraphics[width=.9\linewidth]{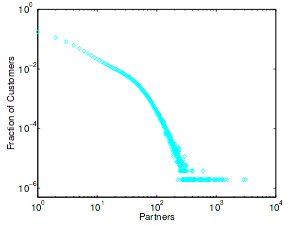}
		\caption{\DD of the number of unique callers and callees from a mobile phone operator.~\cite{seshadri2008mobile}}
		\label{fig:1}
	\end{subfigure}
	\begin{subfigure}{.24\textwidth}
		\centering
		\includegraphics[width=.9\linewidth]{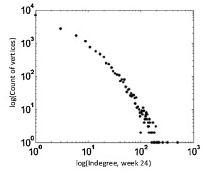}
		\caption{In-\DD between shop-to-shop recommendations from an online marketplace.~\cite{stephen2009explaining}}
		\label{fig:2}
	\end{subfigure}
	\begin{subfigure}{.24\textwidth}
		\centering
		\includegraphics[width=.9\linewidth]{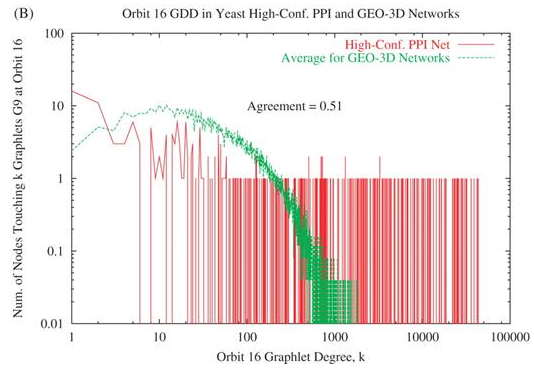}
		\caption{Graphlet \DD from a biological model.~\cite{prvzulj2007biological} \\ \\}
		\label{fig:3}
	\end{subfigure}
	\begin{subfigure}{.24\textwidth}
		\centering
		\includegraphics[width=.9\linewidth]{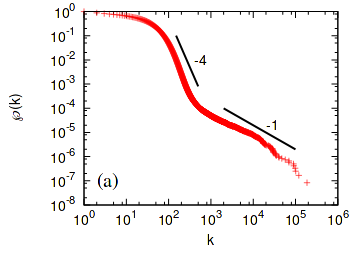}
		\caption{\DD of users of Cyworld, the largest online social network of South Korea.~\cite{ahn2007analysis} \\}
		\label{fig:4}
	\end{subfigure}
	\begin{subfigure}{.24\textwidth}
		\centering
		\includegraphics[width=.9\linewidth]{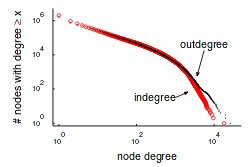}
		\caption{\DDs of users of Flickr, an online social network.~\cite{cha2009measurement} \\}
		\label{fig:5}
	\end{subfigure}
	\begin{subfigure}{.24\textwidth}
		\centering
		\includegraphics[width=.9\linewidth]{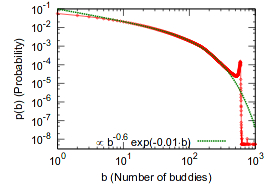}
		\caption{\DD of the length of the contact list in Microsoft Messenger network.~\cite{leskovec2008planetary}}
		\label{fig:7}
	\end{subfigure}
	\begin{subfigure}{.24\textwidth}
		\centering
		\includegraphics[width=.9\linewidth]{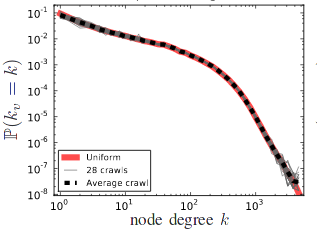}
		\caption{\DD of the number of friends from FaceBook, a social network.~\cite{gjoka2010walking}}
		\label{fig:8}
	\end{subfigure}
	\begin{subfigure}{.24\textwidth}
		\centering
		\includegraphics[width=.9\linewidth]{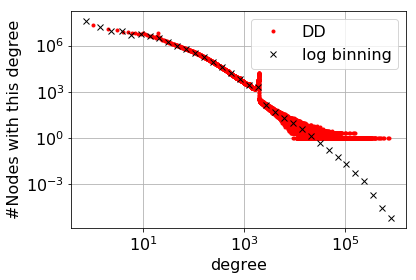}
		\caption{Out-\DD of the number of followees on Twitter.~\cite{trolliet2020interest} \\}
		\label{fig:9}
	\end{subfigure}
	\caption{\DDs extracted from different seminal papers studying networks from various domains.}
	\label{fig:DD_from_everywhere}
	\vspace{-0.6cm}
\end{figure}

Complex networks appear in the empirical study of real world networks from various domains, such that social, biology, economy, technology, ...
Most of those networks exhibit common properties, such as high clustering coefficient, communities, ... Probably the most studied of those properties is the \DDtot (named \DD in the rest of the paper), which is often observed as following a power-law distribution.
Random network models have thus focused on being able to build graphs exhibiting power-law \DDs, such as the well-known Barabasi-Albert model~\cite{albert2002statistical} or the Chung-Lu model~\cite{chung2006complex}, but also models for directed networks~\cite{bollobas2003directed} or for networks with communities~\cite{sallaberry2013model}. However, this is common to find real networks with \DDs not perfectly following a power-law. For instance for social networks, Facebook has been shown to follow a broken power-law\footnote{We call a broken power-law a concatenation of two power-laws, as defined in~\cite{johannesson2006afterglow}.}~\cite{gjoka2010walking}, while Twitter only has the distribution tail following a power-law and some atypical behaviors due to Twitter's policies, as we report in Section~\ref{sec:Twitter_DD}.

It is yet crucial to build models able to reproduce the properties of real networks. Indeed, some studies such as fake news propagation or evolution over time of the networks cannot always be done empirically, for technical or ethical reasons. Carrying out simulations with random networks created with well-built models is a solution to study real networks without directly experimenting on them. Those models have to create networks with similar properties as real ones, while staying as simple as possible.

In this paper, we propose a random growth model able to create graphs with almost any (under some conditions) given \DD. Classical models usually choose the nodes receiving new edges proportionally to a linear attachment function $f(i)=i$ (or $f(i)=i+b$)~\cite{albert2002statistical,bollobas2003directed}. The theoretical \DD of the networks generated by those models is computed using a recurrence equation. The main idea of this paper is to reverse this recurrence equation to express the attachment function $f$ as a function of the \DD. This way, for a given \DD, we can compute the associated attachment function, and use it in a proposed random growth model to create graphs with the wanted \DD. The given \DD can either be theoretical, or extracted from a real network.

We compute the attachment function associated with some classical \DD, homogeneous ones such as Poisson or geometric distributions, and heterogeneous ones such as exact power-law and broken power-law. We also study the undirected \DD of a Twitter snapshot of 400 million nodes and 23 billion edges, extracted by Gabielkov et al.~\cite{gabielkov2012complete} and made available by the authors. We notice it has an atypical shape, due to Twitter's policies. We compute empirically the associated attachment function, and use the model to build random graphs with this \DD. A necessary condition is that the given \DD must be defined for all degrees under the (arbitrary chosen) maximum value. However this condition can be circumvented doing an interpolation between existing points to estimate the missing ones, as discussed in Section~\ref{sec:Real_DD_Twitter}.

Finally, we study in Section~\ref{sec:PA_HT} some connections between attachment functions and probability distributions. More precisely, we show that in our model, unless for some really unusual cases, the probability distribution will be heavy-tailed if and only if the attachment function diverges.

The rest of the paper is organized as follows. We first discuss the related work in Section~\ref{sec:RW}. In Section~\ref{sec:Model}, we present the new model, and invert the recurrence equation to find the relation between the attachment function and the \DD. We apply this relation to compute the attachment function associated to a power-law \DD, a broken-power law \DD, and other theoretical distributions. In Section~\ref{sec:Real_DD_Twitter} we apply our model on a real-world \DD, the undirected \DD of Twitter. We finally show the link between the divergence of the attachment function and the heavy-tailed property of the probability distribution in Section~\ref{sec:PA_HT}.

\renewcommand{\thefootnote}{\arabic{footnote}}

\vspace{-0.4cm}
\section{Related Work}
\label{sec:RW}
\vspace{-0.4cm}

The \DDtot has been computed for a lot of networks, in particular for social networks such as Facebook~\cite{gjoka2010walking} or Microsoft Messenger~\cite{leskovec2008planetary}.
Note that Myers et al. have also studied \DDs for Twitter in~\cite{myers2014information}, using a different dataset than the one of~\cite{gabielkov2012complete}. 

Questioning the relevance of power-law fits is not new: for instance, Clauset et al.~\cite{clauset2009power} or Lima-Mendez and van Helden~\cite{lima2009powerful} have already deeply questioned the myth of power-law -as Lima-Mendez and van Helden call it-, and develop tools to verify if a distribution can be considered as a power-law or not. 
Clauset et al. apply the developed tools on 24 distributions extracted from various domains of literature, which have all been considered to be power-laws. Among them, ``17 of the 24 data sets are consistent with a power-law distribution", and ``there is only one case in which the power law appears to be truly convincing, in the sense that it is an excellent fit to the data and none of the alternatives carries any weight".
In the continuity of this work, Broido and Clauset study in~\cite{broido2019scale} the \DD of nearly 1000 networks from various domains, and conclude that ``fewer than 36 networks (4\%) exhibit the strongest level of evidence for scale-free structure``.

The study of Clauset et al.~\cite{clauset2009power} only considered distributions which have a power-law shape when looking at the distribution in log-log. As a complement, we gathered \DDs from literature which clearly do not follow power-law distributions to show their diversity. 
We extracted from literature \DDs of networks from various domains: biology, economy, computer science, ... Each presented \DD comes from a seminal well cited paper of the respective domains. They are gathered in Figure~\ref{fig:DD_from_everywhere}. Various shapes can be observed from those \DDs, which could (by eyes) be associated with exponential (Fig.~\ref{fig:2},~\ref{fig:3}), broken power-law (Fig.~\ref{fig:1},~\ref{fig:5},~\ref{fig:8}), or even some kind of inverted broken power-law (Fig~\ref{fig:4}). We also observe \DDs with specific behaviors (Fig.~\ref{fig:7},~\ref{fig:9}). 

The first proposed models of random networks, such as the Erdős–Rényi model~\cite{erdHos1960evolution}, build networks with a homogeneous \DD. The observation that a lot of real-world networks follow power-law \DDs lead Albert and Barabasi to propose their famous model with linear preferential attachment~\cite{albert2002statistical}. It has been followed by a lot of random growth models, e.g.~\cite{bollobas2003directed,chung2006complex} also giving a \DD in power-law.
A few models permit to build networks with any \DD: for instance, the configuration model~\cite{bollobas1980probabilistic,newman2001random} takes as parameter a \DD $P$ and a number of nodes $n$, creates $n$ nodes with a degree randomly picked following $P$, then randomly connects the half-edges of every node. Goshal and Newman propose in~\cite{ghoshal2007growing} a model generating non-growing networks (where, at each time-step, a node is added and another is deleted) which can achieve any \DD, using a method close to the one proposed in this paper.
However, both of those models generate non-growing networks, while most real-world networks are constantly growing.

\section{Presentation of the model}
\label{sec:Model}

The proposed model is a generalization of the model introduced by Chung and Lu in~\cite{chung2006complex}. 
At each time step, we have either a node event or an edge event. 
During a node event, a node is added with an edge attached to it; during an edge event, an edge is added between two existing nodes. Each node to which the edge is connected is randomly chosen among all nodes with a probability proportional to a given function $f$, called the \textit{attachment function}.
The model is as follows: 
\begin{itemize}
	\item[$\triangleright$] We start with an initial graph $G_0$.
	\item[$\triangleright$] At each time step t:
	\begin{itemize}
		\item[-] With probability $p$: we add a node $u$, and an edge $(u,v)$ where the node $v$ is chosen randomly between all existing nodes with a probability $\frac{f(deg(v))}{\sum_{w \in V} f(deg(w))}$;
		\item[-] With probability $(1-p)$: we add an edge $(u,v)$ where the nodes $u$ and $v$ are chosen randomly between all existing nodes with a probability $\frac{f(deg(u))}{\sum_{w \in V} f(deg(w))}$ and $\frac{f(deg(v))}{\sum_{w \in V} f(deg(w))}$.
	\end{itemize}
\end{itemize}
Note that the Chung-Lu model is the particular case where $f(i)=i$ for all $i \geq 1$.
We call \textit{generalized Chung-Lu model} the proposed model where $f(i)= i+b$, for all $i \geq 1$ with $b > -1$.

\subsection{Inversion of the recurrence equation}

\renewcommand{\thefootnote}{\fnsymbol{footnote}}

The common way to find the \DD of classical random growth models is to study the recurrence equation of the evolution of the number of nodes with degree $i$ between two time steps. This equation can sometimes be easily solved, sometimes not. But what matters for us is that the common process is to start from a given model -thus an attachment function $f$-, and use the recurrence equation to find the \DD $P$.
In this section, we show that the recurrence equation of the proposed model can be reversed such that, if $P$ if given, we can find an associated attachment function $f$.
\\
\begin{theorem}
	In the proposed model, 
		if $P$ satisfies the following condition\footnote{We conjecture this condition is verified for any probability distribution $P$ with finite mean, though we didn't succeed to prove it.}: 
		\[
		\exists K / \forall i \geq 1, | \frac{P(k>i+1)}{P(i+1)} - \frac{P(k>i)}{P(i)} | \leqslant K
		\]
		, then
	if the attachment function is chosen as: 
	\begin{equation}
		\forall i \geq 1, f(i) = \frac{1}{P(i)} \sum \limits_{k=i+1}^{\infty} P(k),
		\label{eq:thm}
	\end{equation}
	then the \DD of the created graph is distributed according to $P$.\footnote{Note that Equation~\ref{thm:recurrence_equation} can also be expressed as $f(i)=\frac{P(k>i)}{P(i)}$.}
	\label{thm:recurrence_equation}
\end{theorem}


\begin{ComplexNet}
	\begin{proof}
		We consider the variation of the number of nodes of degree $i$ $N(i,t)$ between a time step from t to (t+1). 
		During this time step, a node with degree $i$ may gain a degree and thus diminishes by 1 the number of nodes of degree $i$. This happens with a probability $p+2(1-p)$ (the mean number of half-edges connected to existing nodes during a time step) $\times \frac{f(i)}{\sum_{j \geq 1} f(j)N(j,t)}$ (the probability for this particular node of degree $i$ to be chosen). Since it is the same for all nodes of degree $i$, the number of nodes going from degree $i$ to $i+1$ during a time step is $\big( p+2(1-p) \big) \times \frac{f(i)}{\sum_{j \geq 1} f(j)N(j,t)} \times N(i,t)$. In the same way, some nodes with degree $i-1$ may be connected to an edge and increase the number of nodes of degree $i$. Finally, with probability $p$, a node of degree 1 is added. Gathering those contributions, taking the expectation, and using concentration results give the following equation: 
		\begin{align}
			&\mathbb{E}[N(i,t+1)] - \mathbb{E}[N(i,t)] = \\
			&p \delta_{i,1} + (2-p)\frac{f(i-1) }{\sum \limits_{j \geq 1} f(j)\mathbb{E}[N(j,t)]} \mathbb{E}[N(i-1,t)] - (2-p)\frac{f(i) }{\sum \limits_{j \geq 1} f(j)\mathbb{E}[N(j,t)]} \mathbb{E}[N(i,t)] \nonumber
		\end{align}
		where $\delta_{i,j}$ is the Kronecker delta.
		The first term of the right hand is the probability of addition of a node. The second (resp. third) term is the probability that a node of degree $i-1$ (resp. $i$) gets chosen to be the end of an edge. The factor $(2-p)=p+2(1-p)$ comes from the fact that this happens with probability $p$ during a node event (connection of a single half-edge) and with probability $2(1-p)$ during an edge event (possible connection of 2 half-edges). 
		
		Let $P(i)=\underset{t \rightarrow +\infty}{lim} \frac{\mathbb{E}[N(i,t)]}{pt}$ (the $p$ in the denominator comes from the fact that $\mathbb{E}[N(t)]=pt$). 
		We denote $g(i)= \frac{2-p}{p} \frac{f(i)}{\sum_{j \geq 1} f(j)P(j)}$. We first show that $g(i)= \frac{1}{P(i)} \sum \limits_{k=i+1}^{\infty} P(k)$. We will then show that we can choose $f=g$.\\
		
		We use the following lemma from~\cite{chung2006complex}:
		\begin{lemma}
			Let $(a_t)$, $(b_t)$, $(c_t)$ be three sequences such that $a_{t+1}=(1-\frac{b_t}{t})a_t+c_t$, $\underset{t \rightarrow +\infty}{\lim} b_t = b>0$, and $\underset{t \rightarrow +\infty}{\lim} c_t = c$. Then $\underset{t \rightarrow +\infty}{\lim} \frac{a_t}{t}$ exists and equals $\frac{c}{1+b}$.
			\label{lemma:lim_rec}
		\end{lemma}
		
		For $i=1$, the equation becomes:
		\begin{equation}
			\mathbb{E}[N(1,t+1)] - \mathbb{E}[N(1,t)] = p - (2-p)\frac{f(1)}{\sum \limits_{j \geq 1} f(j)\mathbb{E}[N(j,t)]} \mathbb{E}[N(1, t)].
		\end{equation}
		
		Taking $a_t=\frac{\mathbb{E}[N(1,t)]}{p}$, $b_t = \frac{(2-p)f(1)}{p \sum_{j \geq 1} f(j) \frac{\mathbb{E}[N(j,t)]}{pt}}$, and $c_t = 1$,
		we have $\underset{t \rightarrow +\infty}{\lim} b_t = g(1)>0$ and $\underset{t \rightarrow +\infty}{\lim} c_t = 1$. We can thus apply Lemma~\ref{lemma:lim_rec}:
		\begin{equation}
			\underset{t \rightarrow +\infty}{\lim} \frac{\mathbb{E}[N(1,t)]}{pt} = P(1) = \frac{1}{1+g(1)}.
		\end{equation}
		
		Now, $\forall i \geq 2$, taking $a_t=\frac{\mathbb{E}[N(i,t)]}{p}$, $b_t = \frac{(2-p)f(i)}{p \sum_{j \geq 1} f(j) \frac{\mathbb{E}[N(j,t)]}{pt}}$, and
		\\ 
		$c_t = \frac{(2-p)f(i-1)}{p \sum_{j \geq 1} f(j) \frac{\mathbb{E}[N(j,t)]}{pt}} \frac{\mathbb{E}[N(i-1,t)]}{pt}$,
		we have $\underset{t \rightarrow +\infty}{\lim} b_t = g(i)>0$ and $\underset{t \rightarrow +\infty}{\lim} c_t = g(i-1)P(i-1)$. Lemma~\ref{lemma:lim_rec} gives:
		\begin{equation}
			\underset{t \rightarrow +\infty}{\lim} \frac{\mathbb{E}[N(i,t)]}{pt} = P(i) = \frac{g(i-1)P(i-1)}{1+g(i)}.
			\label{eq:rec_eq_from_proof}
		\end{equation}
		
		\noindent
		Iterating over Equation~\ref{eq:rec_eq_from_proof}, we express $g$ as a function of $P$:
		\begin{equation*}
			g(i)P(i) = g(i-1) P(i-1) - P(i) \nonumber 
			= g(1)P(1) - \sum \limits_{k=2}^{i} P(k) \nonumber 
			= 1 - \sum \limits_{k=1}^{i} P(k) \nonumber \\ 
		\end{equation*}
		\begin{equation}
			\implies g(i) = \frac{1}{P(i)} \sum \limits_{k=i+1}^{\infty} P(k)
		\end{equation}
		\\
		Now, notice that:
		\begin{equation}
			\sum \limits_{k = 1}^{\infty} g(k) P(k) = \sum \limits_{k = 1}^{\infty} \frac{2-p}{p} \frac{f(k)}{\sum_{k' = 1}^{\infty} f(k') P(k')} P(k) = \frac{(2-p)}{p}.
		\end{equation}
		So $g(i)$ satisfies $g(i)=\frac{2-p}{p} \frac{g(i)}{\sum_{k = 1}^{\infty} g(k) P(k)}$. Hence the attachment function can be chosen as $f=g$, which concludes the proof.\qed
	\end{proof}
\end{ComplexNet}


\begin{longversion}
	Let $N(i,t)$ be the random variable corresponding to the number of nodes of degree $i$ at time $t$ in the graph, $N(t)$ the total number of nodes at time $t$, and $P(i)=\underset{t \rightarrow +\infty}{lim} \E[\frac{N(i,t)}{N(t)}]$ the probability that a random node has degree $i$ in the asymptotic \DD.
	
	Before proving Theorem~\ref{thm:recurrence_equation}, we need some results on the concentration of $N(t)$ and $\sum_{j \geq 1} f(j) N(j,t)$.
	We start with $N(t)$. We will need the following lemma:
	
	\begin{lemma}[Chernoff bounds, consult Chapter 4.2 in \cite{MU_book}] \label{lemma:Chernoff}
		Let $X_1,X_2,\ldots,X_t$ be independent indicator random variables with $\Pr[X_i=1]=p_i$ and $\Pr[X_i=0]=1-p_i$. Let $X = \sum_{i=1}^{t} X_i$ and $\mu = \E[X] = \sum_{i=1}^{t} p_i$. Then
		\[
		\Pr[|X-\mu| \geqslant \delta \mu] \leq 2 e^{-\mu\delta^2/3}.
		\]
	\end{lemma}
	
	$N(t)$ is a random variable following a binomial distribution $N(t) \sim B(t,p) + n_0$, with $n_0$ the number of nodes in the initial graph. We can thus use Lemma~\ref{lemma:Chernoff} on $N(t)$; since $\E[N(t)]=pt$, setting $\delta = \sqrt{\frac{9 \ln{t}}{pt}}$ we get:
	
	\begin{corollary} \label{cor:Vt_concentr}
		\begin{equation}
			\Pr[|N(t)-pt| \geqslant \sqrt{9 p t \ln{t}}] \leq 2/t^3.
		\end{equation}
	\end{corollary}
	
	We also have the following result on $P$:
	\begin{lemma} \label{lem:P_concentrated}
		$P(i) \underset{t \rightarrow +\infty}{\sim} \frac{\E[N(i,t)]}{pt}$
	\end{lemma}
	\begin{proof}
		Using Corollary~\ref{cor:Vt_concentr}, we have:
		\begin{align}
			\E\left[\frac{N(i,t)}{N(t)}\right] &= \frac{\E[N(i,t)]}{pt \pm \sqrt{9 p t \ln{t}} + \mathcal{O}\left(\frac{2}{t^3}\right)} \\
			&= \left(1-\frac{2}{t^3}\right) \frac{\E[N(i,t)]}{pt \pm \sqrt{9 p t \ln{t}}} + \mathcal{O}\left(\frac{2}{t^3}\right) \\
			&\hspace{-0.3cm} \underset{t \rightarrow +\infty}{\sim} \frac{\E[N(i,t)]}{pt},
		\end{align}
		where we used the fact that $\E[N(t)]=pt$ and $\E[N(i,t)] \leqslant t$.
		\qed
	\end{proof}
	
	We now discuss the concentration of $\sum_{j \geq 1} f(j) N(j,t)$. Let us define 
	\[
	Z_t=\sum_{j \geq 1} f(j) N(j,t).
	\]
	Using the following lemma from~\cite{Hoeff}:
	\begin{lemma} [Hoeffding's inequality, \cite{Hoeff}] \label{lemma:hoeff}
		Let $X_1, X_2, \ldots, X_t$ be independent random variables such that $\Pr[X_k \in [a_k,b_k]] =1$. Let $X = \sum_{k=1}^{t}X_k$. Then
		\[
		\Pr[|X - \E[X]| \geqslant \delta] \leqslant 2 \exp\left\{-\frac{2 \delta^2}{\sum_{k=1}^t(a_k-b_k)^2}\right\}.
		\]
	\end{lemma}
	We can show that:
	\begin{lemma} \label{lemma:Zt_concentr}
		If the following condition is satisfied:
		\[
		\exists K / \forall i \geqslant 1, |f(i+1)-f(i)| \leqslant K
		\]
		Then:
		\[
		\Pr[|Z_t - \E[Z_t]| \geqslant \sqrt{32 K^2 t \ln{t}}] = \mathcal{O}\left(\frac{1}{t^4}\right).
		\]
	\end{lemma}
	\begin{proof}
		First, remind that $Z_t$ can either be express as $Z_t=\sum_{j \geq 1} f(j) N(j,t)$ or $Z_t = \sum_{u \in V_t} f(deg_t(u))$, with $deg_t(u)$ the degree of node $u$ at time $t$.
		But $Z_t$ can also be express as a sum of independent random variables $X_1+X_2+...+X_t$, with $X_k$ the variation of $Z_k$ during the time step $k$, i.e. $X_k=Z_k-Z_{k-1}$.
		In practice, $X_k$ can take those different values:
		\begin{itemize}
			\item With probability $p$, a node and an edge are added to the graph, and $X_k = f(deg_k(u)+1)-f(deg_k(u)) + f(1)$, with $u$ the chosen node at time step $k$;
			\item With probability $(1-p)$, an edge is added between two existing nodes, and $X_k = f(deg_k(u)+1)-f(deg_k(u)) + f(deg_k(v)+1)-f(deg_k(v))$, with $u$ and $v$ the chosen nodes.
		\end{itemize}
		Using the condition on $f$, we see that we can bound $X_k$ by $-2K \leqslant X_k \leqslant 2K$.
		\\
		We can thus apply Lemma~\ref{lemma:hoeff} with $X =\sum_{k=1}^{t}X_k =Z_t$, $a_i=-2K$ and $b_i=2K$ to obtain:
		\begin{equation}
			\Pr[|Z_t - \E[Z_t]| \geqslant \delta] \leqslant 2 \exp\left\{-\frac{2 \delta^2}{t (4K)^2}\right\}.
		\end{equation}
		Now, setting $\delta = \sqrt{32 K^2 t \ln{t}}$ we get:
		\[
		\Pr[|Z_t - \E[Z_t]| \geqslant \sqrt{32 K^2 t \ln{t}}] \leqslant 2 \exp\left\{-\frac{2 \cdot 32 K^2 t \ln{t}}{t (4K)^2}\right\} = \mathcal{O}\left(\frac{1}{t^4}\right).
		\]
		\qed
	\end{proof}
	
	We will finally need the following lemma from~\cite{ChLu_book}:
	\begin{lemma}[Compare Chapter 3.3 in \cite{ChLu_book}]
		Let $(a_t)$, $(b_t)$, $(c_t)$ be three sequences such that $a_{t+1}=(1-\frac{b_t}{t})a_t+c_t$, $\underset{t \rightarrow +\infty}{\lim} b_t = b>0$, and $\underset{t \rightarrow +\infty}{\lim} c_t = c$. Then $\underset{t \rightarrow +\infty}{\lim} \frac{a_t}{t}$ exists and equals $\frac{c}{1+b}$.
		\label{lemma:lim_rec}
	\end{lemma}

	We are now ready to prove Theorem~\ref{thm:recurrence_equation}.
	\begin{proof}[Proof of Theorem \ref{thm:recurrence_equation}]
			During the proof, we will consider the following conditions as true: 
			\begin{itemize}
				\item[C1)] $\exists K / \forall i \geqslant 1, |f(i+1)-f(i)| \leqslant K$,
				\item[C2)] $\sum_{j \geq 1} f(j)P(j) = \mu$, $\mu \in \mathbb{R}^*_+$.
			\end{itemize}
			where we remind that $P$ is defined as $P(i)=\underset{t \rightarrow +\infty}{lim} \E[\frac{N(i,t)}{N(t)}]$.
			We will verify at the end of the proof that the first condition is equivalent to the condition of Theorem~\ref{thm:recurrence_equation}, and the second condition is indeed satisfied for the chosen $f$ . 
		
		We consider the variation of the number of nodes of degree $i$ $N(i,t)$ between a time step from t to (t+1). 
		During this time step, a node with degree $i-1$ may gain a degree and thus increases by 1 the number of nodes of degree $i$. This happens with a probability $p+2(1-p)$ (the mean number of half-edges connected to existing nodes during a time step) $\times \frac{f(i-1)}{\sum_{j \geq 1} f(j)N(j,t)}$ (the probability for this particular node of degree $i-1$ to be chosen). Since it is the same for all nodes of degree $i-1$, the number of nodes going from degree $i-1$ to $i$ during a time step is $\big( p+2(1-p) \big) \times \frac{f(i-1)}{\sum_{j \geq 1} f(j)N(j,t)} \times N(i-1,t)$. In the same way, a node with degree $i$ may be connected to an edge, thus becoming a node with degree $i+1$ and decreasing the number of nodes of degree $i$. Finally, with probability $p$, a node of degree 1 is added. 
		Gathering those contributions give the following equation: 
		\begin{align} \label{eq:ME}
			&N(i,t+1) - N(i,t) = \\
			&p \delta_{i,1} + (2-p)\frac{f(i-1) }{\sum \limits_{j \geq 1} f(j)N(j,t)} N(i-1,t) - (2-p)\frac{f(i) }{\sum \limits_{j \geq 1} f(j)N(j,t)} N(i,t) \nonumber
		\end{align}
		where $\delta_{i,j}$ is the Kronecker delta.
		The first term of the right hand is the probability of addition of a node. The second (resp. third) term is the probability that a node of degree $i-1$ (resp. $i$) gets chosen to be the end of an edge. The factor $(2-p)=p+2(1-p)$ comes from the fact that this happens with probability $p$ during a node event (connection of a single half-edge) and with probability $2(1-p)$ during an edge event (possible connection of 2 half-edges). 
		
		We take the expectation on both sides and use Lemma~\ref{lemma:Zt_concentr} to obtain:
		\begin{align} \label{eq:ME_expectation}
			&\mathbb{E}[N(i,t+1)] - \mathbb{E}[N(i,t)] = p \delta_{i,1} \\
			&+ (2-p)\frac{f(i-1) }{\sum \limits_{j \geq 1} f(j)\mathbb{E}[N(j,t)] + \mathcal{O}\left(\sqrt{t \ln{t}}\right)} \mathbb{E}[N(i-1,t)] \\ \nonumber
			&- (2-p)\frac{f(i) }{\sum \limits_{j \geq 1} f(j)\mathbb{E}[N(j,t)] + \mathcal{O}\left(\sqrt{t \ln{t}}\right)} \mathbb{E}[N(i,t)] \nonumber
		\end{align}
		
		We denote $g(i)= \frac{2-p}{p} \frac{f(i)}{\sum_{j \geq 1} f(j)P(j)}$. We first show that $g(i)= \frac{1}{P(i)} \sum \limits_{k=i+1}^{\infty} P(k)$. We will then show that we can choose $f=g$.
		\\
		For $i=1$, Equation~\ref{eq:ME_expectation} becomes:
		\begin{equation}
			\mathbb{E}[N(1,t+1)] - \mathbb{E}[N(1,t)] = p - (2-p)\frac{f(1)}{\sum \limits_{j \geq 1} f(j)\mathbb{E}[N(j,t)] + \mathcal{O}\left(\sqrt{t \ln{t}}\right)} \mathbb{E}[N(1, t)].
		\end{equation}
		
		Taking:
		\[
		a_t=\frac{\mathbb{E}[N(1,t)]}{p},
		\]
		\[
		b_t = \frac{(2-p)f(1)}{p \sum_{j \geq 1} f(j) \frac{\mathbb{E}[N(j,t)]}{pt} + \mathcal{O}\left(\sqrt{\frac{\ln{t}}{t}}\right)},
		\] 
		\[
		c_t = 1,
		\]
		we have $\underset{t \rightarrow +\infty}{\lim} b_t = g(1)>0$ and $\underset{t \rightarrow +\infty}{\lim} c_t = 1$. We can thus apply Lemma~\ref{lemma:lim_rec} (and use Lemma~\ref{lem:P_concentrated} to recognize $P(1)$):
		\begin{equation} \label{eq:relation-P1-g1}
			\underset{t \rightarrow +\infty}{\lim} \frac{\mathbb{E}[N(1,t)]}{pt} = P(1) = \frac{1}{1+g(1)},
		\end{equation}
		Now, $\forall i \geq 2$, taking:
		\[
		a_t=\frac{\mathbb{E}[N(i,t)]}{p},
		\] 
		\[
		b_t = \frac{(2-p)f(i)}{p \sum_{j \geq 1} f(j) \frac{\mathbb{E}[N(j,t)]}{pt}  + \mathcal{O}\left(\sqrt{\frac{\ln{t}}{t}}\right)},
		\]
		\[
		c_t = \frac{(2-p)f(i-1)}{p \sum_{j \geq 1} f(j) \frac{\mathbb{E}[N(j,t)]}{pt}  + \mathcal{O}\left(\sqrt{\frac{\ln{t}}{t}}\right)} \frac{\mathbb{E}[N(i-1,t)]}{pt},
		\]
		we have $\underset{t \rightarrow +\infty}{\lim} b_t = g(i)>0$ and $\underset{t \rightarrow +\infty}{\lim} c_t = g(i-1)P(i-1)$. Lemma~\ref{lemma:lim_rec} and Lemma~\ref{lem:P_concentrated} give:
		\begin{equation}
			\underset{t \rightarrow +\infty}{\lim} \frac{\mathbb{E}[N(i,t)]}{pt} = P(i) = \frac{g(i-1)P(i-1)}{1+g(i)}.
			\label{eq:rec_eq_from_proof}
		\end{equation}
		
		\noindent
		Iterating over Equation~\ref{eq:rec_eq_from_proof}, we express $g$ as a function of $P$:
		\begin{align*}
			g(i)P(i) &= g(i-1) P(i-1) - P(i) \nonumber \\ 
			&= g(i-2) P(i-2) - P(i-1) - P(i) \nonumber \\
			&= \cdots \nonumber \\ 
			&= g(1)P(1) - \sum \limits_{k=2}^{i} P(k) \nonumber \\ 
			&= 1 - \sum \limits_{k=1}^{i} P(k) \nonumber \\ 
		\end{align*}
		\begin{equation}
			\implies g(i) = \frac{1}{P(i)} \sum \limits_{k=i+1}^{\infty} P(k)
		\end{equation}
		where we used Equation~\ref{eq:relation-P1-g1} to replace $g(1)P(1)$.
		\\
		Now, notice that:
		\begin{equation} \label{eq:gk=cst}
			\sum \limits_{k = 1}^{\infty} g(k) P(k) = \sum \limits_{k = 1}^{\infty} \frac{2-p}{p} \frac{f(k)}{\sum_{k' = 1}^{\infty} f(k') P(k')} P(k) = \frac{(2-p)}{p}.
		\end{equation}
		So $g(i)$ satisfies $g(i)=\frac{2-p}{p} \frac{g(i)}{\sum_{k = 1}^{\infty} g(k) P(k)}$. Hence the attachment function can be chosen as $f=g$.
		\\
			We finally have to verify the conditions we put at the beginning of the proof are true. 			
			The first condition is equivalent to the condition of the theorem for the given $f$.
			The second condition is given by Equation~\ref{eq:gk=cst}, which conclude the proof.
		\qed
	\end{proof}
\end{longversion}


For a given probability law, Theorem~\ref{thm:recurrence_equation} can be used to compute the attachment function which, when used in the model, will give this probability law as DD.

\;
With the presented model, we also have an implicit constraint between the mean degree and the parameter p. 
Indeed by construction, we have $\mathbb{E}[N(t)]=pt$ and $\mathbb{E}(|E|(t))=t$ with $|E|(t)$ the number of edges at time $t$, leading to a mean-degree of $\frac{1}{p}$. But the mean-degree can also be expressed as $\sum_{k \geq 1} k P(k)$. 
\begin{condition}
	The parameter $p$ has to satisfy:
	\begin{equation}
		\frac{1}{p} = <k>
	\end{equation}
	\label{cond:condition_mean}
\end{condition}

\noindent
We can finally combine the previous results and present the method to build a random network with a fixed \DD:

\begin{itemize}
	\item[1)] Use Equation~\ref{eq:thm} to compute $f$ from $P$;
	\item[2)] Compute $p$ using Condition~\ref{cond:condition_mean};
	\item[3)] Build the graph with the proposed model, given $(f,p)$ as parameters.
\end{itemize}

\section{Application to some distributions}
\label{sec:Application}

We now apply Equation~\ref{eq:thm} to compute the attachment function for some classical distributions. We first start in Section~\ref{sec:GCL} from the distribution obtained with the generalized Chung-Lu model to show we find a linear dependence, as expected. We then compute in Section~\ref{sec:BPL} the associated attachment function of the broken power-law distribution. 
\begin{ComplexNet}
	Using similar computations (which can be found in Report~\cite{giroire2020random}), we computed the attachment function of other classical distributions. 
\end{ComplexNet}
\begin{longversion}
	We finally compute the exact power-law, geometric law, and Poisson law distributions in Sections~\ref{subsec:PL},~\ref{subsec:Geometric} and ~\ref{subsec:Poisson}.
\end{longversion}
Table~\ref{tab:Others} summarizes those results.

\begin{table}[t]
	\centering
	\begin{tabular}{|c|c|c|c|}
		\hline
		Name & P(i) & f(i) & Condition \\
		\hline
		Generalized Chung-Lu & $C \frac{\Gamma(i+b)}{\Gamma(i+b+\alpha)}$ & $\frac{1}{\alpha-1} i + \frac{b}{\alpha-1}$ & $p=\frac{\alpha-2}{\alpha+b-1}$ \\
		\hline
		Exact Power-Law & $\frac{i^{-\alpha}}{\zeta(\alpha)}$ & $\frac{\zeta(\alpha,i+1)}{i^{-\alpha}}$ & $p = \frac{\zeta(\alpha)}{\zeta(\alpha-1)}$ \\
		\hline
		Geometric Law & $q (1-q)^{i-1}$ & $\frac{1-q}{q}$ & $p = q$ \\
		\hline
		Poisson Law & $\frac{1}{e^{\lambda}-1} \frac{\lambda^i}{i!}$ & $e^{\lambda} \frac{\gamma(i+1,\lambda)}{\lambda^i}$ & $p = \frac{1-e^{-\lambda}}{\lambda}$ \\
		\hline
		Broken Power-Law 
		& 
		$   \left\{
		\begin{array}{ll}
			C \frac{\Gamma(i+b_1)}{\Gamma(i+b_1+\alpha_1)} & \mbox{if } i \leq d \\
			C \gamma \frac{\Gamma(i+b_2)}{\Gamma(i+b_2+\alpha_2)} & \mbox{if } i > d \\
		\end{array}
		\right.$ 
		& 
		cf. eq.~\ref{eq:BPL_f_case1}\&~\ref{eq:BPL_f_case2}
		& cf. eq.~\ref{eq:Condition_p_BPL} \\
		\hline
	\end{tabular}
	\caption{Attachment functions $f$ and conditions on $p$ for some classical probability distributions $P$. $\zeta(s)$ is the Riemann zeta function, $\zeta(s,q)$ the Hurwitz zeta function, and $\gamma(a,x)$ is the lower incomplete Gamma function.}
	\label{tab:Others}
\end{table}

\subsection{Preliminary: Generalized Chung-Lu model}
\label{sec:GCL}

As a first example, by taking a power-law \DD, we should be able to find a linear probability distribution for the generalized Chung-Lu model.

In the general Chung-Lu model, we can show that the real \DD is not an exact power-law but a fraction of Gamma function -equivalent to a power-law for high degrees- of the form:

\begin{equation}
	\forall i \geq 1, P(i) = C \frac{\Gamma(i+b)}{\Gamma(i+b+\alpha)} \underset{i \gg 1}{\sim} i^{-\alpha}
\end{equation}
where $C = (\alpha-1) \frac{\Gamma(b+\alpha)}{\Gamma(b+1)}$, and $\alpha>2$. 
The choice of $\alpha$ determines the slope of the \DD, while the choice of $b$ determines the mean-degree of the graph.

\ms

\nitbf{Constraint on p:}
Condition~\ref{cond:condition_mean} gives:
\begin{longversion}
	\begin{align}
		\frac{1}{p} = \sum \limits_{k=1}^{\infty} k P(k) \nonumber
		&= (\alpha-1) \frac{\Gamma(b+\alpha)}{\Gamma(b+1)} \times \frac{\alpha^2 + \alpha (2b-1) + b(b-1)}{(\alpha-2)(\alpha-1)} \frac{\Gamma(b+1)}{\Gamma(\alpha+b+1)}  \nonumber
		\\
		\implies p &= \frac{(\alpha-2)}{\alpha+b-1}
		\label{eq:cond_M_CL}
	\end{align}
\end{longversion}
\begin{ComplexNet}
	\begin{equation}
		\frac{1}{p} = \sum \limits_{k=1}^{\infty} k P(k) 
		= (\alpha-1) \frac{\Gamma(b+\alpha)}{\Gamma(b+1)} \times \frac{\alpha^2 + \alpha (2b-1) + b(b-1)}{(\alpha-2)(\alpha-1)} \frac{\Gamma(b+1)}{\Gamma(\alpha+b+1)}  \nonumber
		\label{eq:cond_M_CL}
	\end{equation}
	\begin{equation}
		\implies p = \frac{(\alpha-2)}{\alpha+b-1}
	\end{equation}
\end{ComplexNet}

\;

\nitbf{Attachment function f:}
Using Theorem~\ref{thm:recurrence_equation}:

\begin{equation}
	f(i) = \frac{1}{P(i)} \sum \limits_{k\geq i+1} P(k)
	= \frac{\Gamma(i+b+\alpha)}{\Gamma(i+b)} \frac{\Gamma(i+b+1)}{(\alpha-1) \Gamma(i+\alpha+b)}
\end{equation}
\begin{equation}
	\implies f(i) = \frac{1}{\alpha-1}i+\frac{b}{\alpha-1}
\end{equation}

\noindent
As expected, we find a linear attachment function. To create a graph with a wanted slope $\alpha$ and mean-degree $p^{-1}$, one only has to choose $\alpha$ as the wanted slope and $b$ following equation~\ref{eq:cond_M_CL}.
In the particular case $b=0$, we recover the Chung-Lu model of~\cite{chung2006complex}, with a slope of $\alpha=2+\frac{p}{2-p}$ as expected.

\subsection{Broken Power-law}
\label{sec:BPL}

We now study the case of a broken power-law, corresponding to the \DD of real world complex networks, as discussed in Section~\ref{sec:RW}.
which was the one we were interested in initially.
We consider a distribution of the form:
\begin{align} 
	P(i) &= 
	\left\{
	\begin{array}{ll}
		C \frac{\Gamma(i+b_1)}{\Gamma(i+b_1+\alpha_1)} & \mbox{if } i \leq d \\
		C \gamma \frac{\Gamma(i+b_2)}{\Gamma(i+b_2+\alpha_2)} & \mbox{if } i > d \\
	\end{array}
	\right.
\end{align}
where $d, b_1, \alpha_1, b_2,$ and $\alpha_2$ are parameters of our distribution such that 
$\alpha_1>2$, $\alpha_2>2$,
$C$ a normalisation constant, and $\gamma$ chosen in order to obtain continuity for $i=d$. As seen in section~\ref{sec:GCL}, the ratio of gamma functions is close to a power-law as soon as $i$ gets large. Hence, this distribution corresponds to two powers-laws, with different slopes, and a switch between the two at the value $d$.

We can easily find the continuity constant $\gamma$, since it verifies:
\begin{equation}
	\frac{\Gamma(d+b_1)}{\Gamma(d+b_1+\alpha_1)} = \gamma \frac{\Gamma(d+b_2)}{\Gamma(d+b_2+\alpha_2)}
	\implies \gamma = \frac{\Gamma(d+b_1) \Gamma(d+b_2+\alpha_2)}{\Gamma(d+b_1+\alpha_1) \Gamma(d+b_2)}.
\end{equation}

\;

\nitbf{Constraints on C and p:}
The value of C can be computed by summing over all degrees:
\begin{align}
	C &= \Big( \sum \limits_{k=1}^\infty P(k) \Big)^{-1} &=  \Big( \frac{1}{\alpha_1-1} \frac{\Gamma(b_1+1)}{\Gamma(\alpha_1+b_1)} + \frac{\Gamma(b_1+d)}{\Gamma(\alpha_1+b_1+d)} \big(\frac{b_2+d}{\alpha_2-1} - \frac{b_1+d}{\alpha_1-1} \big) \Big)^{-1} 
\end{align}

\noindent
Using Condition~\ref{cond:condition_mean}, $p$ is defined by the following equation:
\begin{align}
	\frac{1}{pC} &= \sum \limits_{k=1}^d k \frac{\Gamma(k+b_1)}{\Gamma(k+b_1+\alpha_1)} + \gamma \sum \limits_{k=d+1}^\infty k \frac{\Gamma(k+b_2)}{\Gamma(k+b_2+\alpha_2)} \nonumber
	\\
	&= \frac{\alpha_1^2 + \alpha_1(2b_1-1)+b_1(b_1-1)}{(\alpha_1-2)(\alpha_1-1)} \frac{\Gamma(b_1+1)}{\Gamma(\alpha_1+b_1+1)} 
	\label{eq:Condition_p_BPL}
	\\
	&\hspace{-0.5cm}- \frac{\alpha_1^2 (d+1) + \alpha_1(b_1(d+2)+d^2-1)+b_1(b_1-1)-d(d+1)}{(\alpha_1-2)(\alpha_1-1)} \frac{\Gamma(b_1+d+1)}{\Gamma(\alpha_1+b_1+d+1)} \nonumber
	\\
	&\hspace{-0.5cm}+ \gamma \frac{\alpha_2^2 (d+1) + \alpha_2(b_2(d+2)+d^2-1)+b_2(b_2-1)-d(d+1)}{(\alpha_2-2)(\alpha_2-1)} \frac{\Gamma(b_2+d+1)}{\Gamma(\alpha_2+b_2+d+1)} \nonumber
\end{align}

\;

\nitbf{Attachment function $f$:}
For the computation of the attachment function, we have to distinguish two cases:

\;

\textbf{Case 1: $i \geq d$}
\begin{longversion}
	\begin{align}
		f(i) = \frac{1}{P(i)} \sum \limits_{k=i+1}^\infty P(k) 
		&= \frac{\Gamma(i+b_2+\alpha_2)}{\Gamma(i+b_2)} \sum \limits_{k=i+1}^\infty  \frac{\Gamma(k+b_2)}{\Gamma(k+b_2+\alpha_2)} \nonumber
		\\
		&= \frac{\Gamma(i+b_2+\alpha_2)}{\Gamma(i+b_2)} \frac{1}{\alpha_2-1} \frac{\Gamma(i+b_2+1)}{ \Gamma(i+b_2+\alpha_2)} \nonumber
		\\
		\implies f(i) &= \frac{1}{\alpha_2-1} i + \frac{b_2}{\alpha_2-1}
		\label{eq:BPL_f_case1}
	\end{align}
\end{longversion}
\begin{ComplexNet}
	\begin{equation}
		f(i) = \frac{\Gamma(i+b_2+\alpha_2)}{\Gamma(i+b_2)} \frac{1}{\alpha_2-1} \frac{\Gamma(i+b_2+1)}{ \Gamma(i+b_2+\alpha_2)}
		= \frac{1}{\alpha_2-1} i + \frac{b_2}{\alpha_2-1}
		\label{eq:BPL_f_case1}
	\end{equation}
\end{ComplexNet}
We find a linear attachment function: indeed for $i>d$, we only take into account the second power-law, hence we expect to find the same result than in section~\ref{sec:GCL}.

\;

\textbf{Case 2: $i < d$}
\begin{longversion}
	\begin{align}
		f(i) &= \frac{\Gamma(i+b_1+\alpha_1)}{\Gamma(i+b_1)} \Bigg( \sum \limits_{k=i+1}^d \frac{\Gamma(k+b_1)}{\Gamma(k+b_1+\alpha_1)} + \gamma \sum \limits_{k=d+1}^\infty \frac{\Gamma(k+b_2)}{\Gamma(k+b_2+\alpha_2)} \Bigg) \nonumber
		\\
		&= \frac{\Gamma(i+b_1+\alpha_1)}{\Gamma(i+b_1)} \Bigg( \frac{1}{\alpha_1-1} \big( \frac{\Gamma(i+b_1+1)}{\Gamma(i+\alpha_1+b_1)} - \frac{\Gamma(b_1+d+1)}{\Gamma(b_1+\alpha_1+d)} \big) + \frac{\gamma}{\alpha_2-1} \frac{\Gamma(b_2+d+1)}{\Gamma(b_2+\alpha_1+d)} \Bigg) \nonumber
		\\
		&= \frac{i+b_1}{\alpha_1-1} + \frac{\Gamma(i+b_1+\alpha_1)}{\Gamma(i+b_1)} \Bigg( \frac{d+b_2}{\alpha_2-1} \frac{\Gamma(b_1+d)}{\Gamma(b_1+\alpha_1+d)} - \frac{1}{\alpha_1-1} \frac{\Gamma(b_1+d+1)}{\Gamma(b_1+\alpha_1+d)} \Bigg) \nonumber
		\\
		f(i) &= \frac{i+b_1}{\alpha_1-1} + \frac{\Gamma(i+b_1+\alpha_1) \Gamma(d+b_1)}{\Gamma(i+b_1) \Gamma(d+b_1+\alpha_1)} \Big( \frac{b_2+d}{\alpha_2-1} - \frac{b_1+d}{\alpha_1-1} \Big)
		\label{eq:BPL_f_case2}
	\end{align}
\end{longversion}
\begin{ComplexNet}
	\begin{align}
		f(i) &= \frac{\Gamma(i+b_1+\alpha_1)}{\Gamma(i+b_1)} \Bigg( \sum \limits_{k=i+1}^d \frac{\Gamma(k+b_1)}{\Gamma(k+b_1+\alpha_1)} + \gamma \sum \limits_{k=d+1}^\infty \frac{\Gamma(k+b_2)}{\Gamma(k+b_2+\alpha_2)} \Bigg) \nonumber
		\\
		&= \frac{i+b_1}{\alpha_1-1} + \frac{\Gamma(i+b_1+\alpha_1) \Gamma(d+b_1)}{\Gamma(i+b_1) \Gamma(d+b_1+\alpha_1)} \Big( \frac{b_2+d}{\alpha_2-1} - \frac{b_1+d}{\alpha_1-1} \Big)
		\label{eq:BPL_f_case2}
	\end{align}
\end{ComplexNet}
In this second case, we have a linear part, in addition to a more complicated part. Note that, for $(\alpha_1,b_1) = (\alpha_2,b_2)$, i.e., when the two power-laws are equals, this second term vanishes, letting as expected only the linear part.
Figure~\ref{fig:Broken_PL_f} shows the shape of $f$. We see that, while the second part is linear as discussed before, the first part is sub-linear.

We used this attachment function to build a network using our model. The \DD is shown in Figure~\ref{fig:Broken_PL}: we see we built a random network with a broken power-law distribution as wanted.

\begin{figure}[t]
	\begin{subfigure}{.48\textwidth}
		\centering
		\includegraphics[width=1.\linewidth]{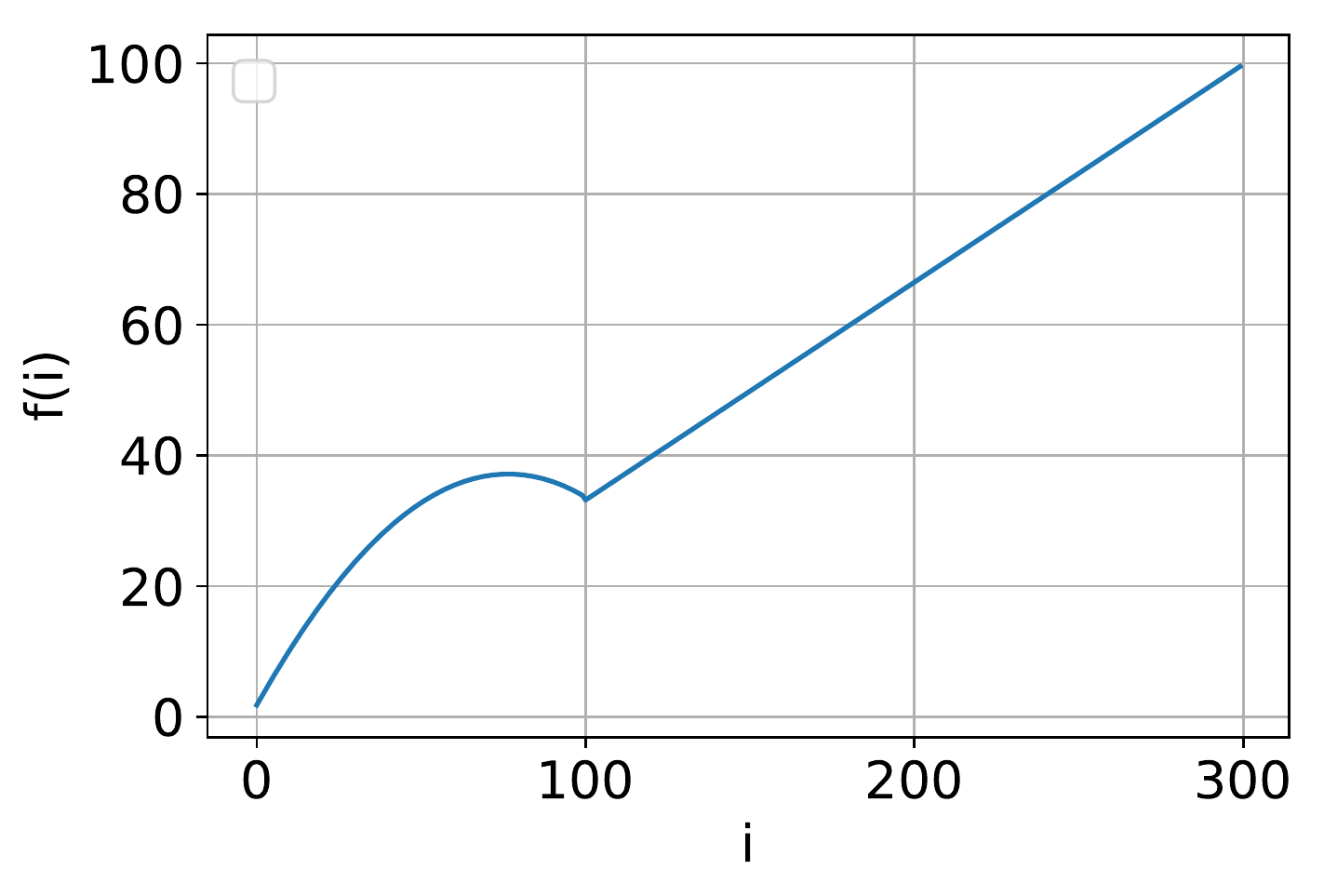}
		\caption{Theoretical attachment function $f$}
		\label{fig:Broken_PL_f}
	\end{subfigure}
	\begin{subfigure}{.48\textwidth}
		\centering
		\includegraphics[width=1.\linewidth]{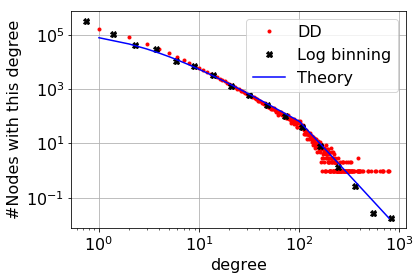}
		\caption{DD of a random network}
		\label{fig:Broken_PL}
	\end{subfigure}
	\caption{Theoretical attachment function $f$ and degree distribution of a random network for the broken power-law distribution. Parameters are $N=5 \cdot 10^{5}$, $b_1=b_2=1$, $\alpha_1=2.1$, $\alpha_2=4$ and $d=100$.}
	\label{fig:fig2}
	\vspace{-0.5cm}
\end{figure}

\begin{longversion}

\subsection{Exact power-law degree distribution}
\label{subsec:PL}

The \DD obtained with the Chun-Lu model -and most of other classical models- gives a power-law only for high degrees.
We can ask ourselves what would be the attachment function associated with an exact power-law degree distribution of the form $P(i)=\frac{i^{-\alpha}}{\zeta(\alpha)}$, where $\zeta(s)=\sum \limits_{k \geq 1} \frac{1}{k^s}$ is the Riemann zeta function.

\subsubsection{Constraints on C and p}

Condition~\ref{cond:condition_mean} gives the following equation:

\[
\frac{1}{p} = \frac{1}{\zeta(\alpha)} \sum \limits_{k=1}^{\infty} k^{1-\alpha} = \frac{\zeta(\alpha-1)}{\zeta(\alpha)}
\]
\begin{equation}
	\implies p = \frac{\zeta(\alpha)}{\zeta(\alpha-1)}
\end{equation}

\subsubsection{Attachment function}

Theorem~\ref{thm:recurrence_equation} gives immediately:

\begin{align}
	f(i) = \frac{1}{P(i)} \sum \limits_{k=i+1}^{\infty} P(k)
	= \frac{\zeta(\alpha,i+1)}{i^{-\alpha}}
\end{align}

\noindent


\subsection{Geometric law}
\label{subsec:Geometric}

We now study the geometric distribution:
\begin{equation}
	\forall i \geq 1, P(i) = q (1-q)^{i-1}
\end{equation}

\subsubsection{Constraints on p}

Condition~\ref{cond:condition_mean} gives:

\begin{equation}
	\frac{1}{p} = \sum \limits_{k \geq 1} k q (1-q)^{k-1}
	= \frac{q}{(1-q)} \frac{(1-q)}{q^2} 
	= \frac{1}{q}
\end{equation}
\begin{equation}
	\implies p = q
\end{equation}

\noindent

\subsubsection{Attachment function}

The attachment function is easy to compute:

\begin{align}
	f(i) = \frac{1}{q (1-q)^{i-1}} \sum \limits_{k \geq i+1} q (1-q)^{k-1} = \frac{1}{(1-q)^i} \frac{(1-q)^{i+1}}{q}
	= \frac{1-q}{q}
\end{align}



\subsection{Poisson law}
\label{subsec:Poisson}

Another classic homogeneous law is the Poisson distribution:

\begin{equation}
	\forall i \geq 1, P(i) = \frac{1}{e^{\lambda}-1} \frac{\lambda^i}{i!}
\end{equation}
The constant $\frac{1}{e^{\lambda}-1}$ has been chosen such that $\sum \limits_{k \geq 1} P(k) = 1$.

\subsubsection{Constraint on p}
The condition~\ref{cond:condition_mean} gives:

\begin{align}
	\frac{1}{p} &= \frac{1}{e^{\lambda}-1} \sum \limits_{k \geq 1} k \frac{\lambda^k}{k!} = \frac{e^\lambda \lambda}{e^\lambda-1}
	\label{eq:condition_m_Poisson}
\end{align}
\begin{equation}
	\implies p = \frac{1-e^{-\lambda}}{\lambda}
\end{equation}

\subsubsection{Attachment function}

Theorem~\ref{thm:recurrence_equation} gives:

\begin{align}
	f(i) = \frac{i!}{\lambda^i} \sum \limits_{k \geq i+1} \frac{\lambda^k}{k!} 
	= \frac{i!}{\lambda^i} \times \frac{e^{\lambda} (i! - \Gamma(i+1, \lambda))}{i!}
	= e^{\lambda} \frac{\gamma(i+1,\lambda)}{\lambda^i}
\end{align}
where $\gamma(a,x)= \int \limits_{t=0}^x t^{a-1} e^{-t} dt$ is the lower incomplete Gamma function.
\\


\end{longversion}

\vspace{-0.3cm}
\section{Real \DDstot}
\label{sec:Real_DD_Twitter}
\vspace{-0.2cm}

The model can also be applied to an empirical \DD. Indeed, we observe in Theorem~\ref{thm:recurrence_equation} that $f(i)$ only depends on the values $P(i)$ which can be arbitrary, that is not following any classical function. This is a good way to model random networks with an atypical \DD.
As an example, we apply our model on the \DD of an undirected version of Twitter, shown as having atypical behavior due to the Twitter policies.
We start with a presentation of this \DD, then apply our model to build a random graph with this distribution.

\begin{figure}[t]
	\begin{subfigure}{.48\textwidth}
		\centering
		\includegraphics[scale=0.35]{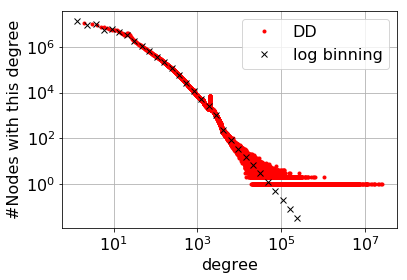}
		\caption{DD of the Twitter's undirected network.}
		\label{fig:Twitter_Undirected_DD}
	\end{subfigure}
	\begin{subfigure}{.48\textwidth}
		\centering
		\includegraphics[scale=0.35]{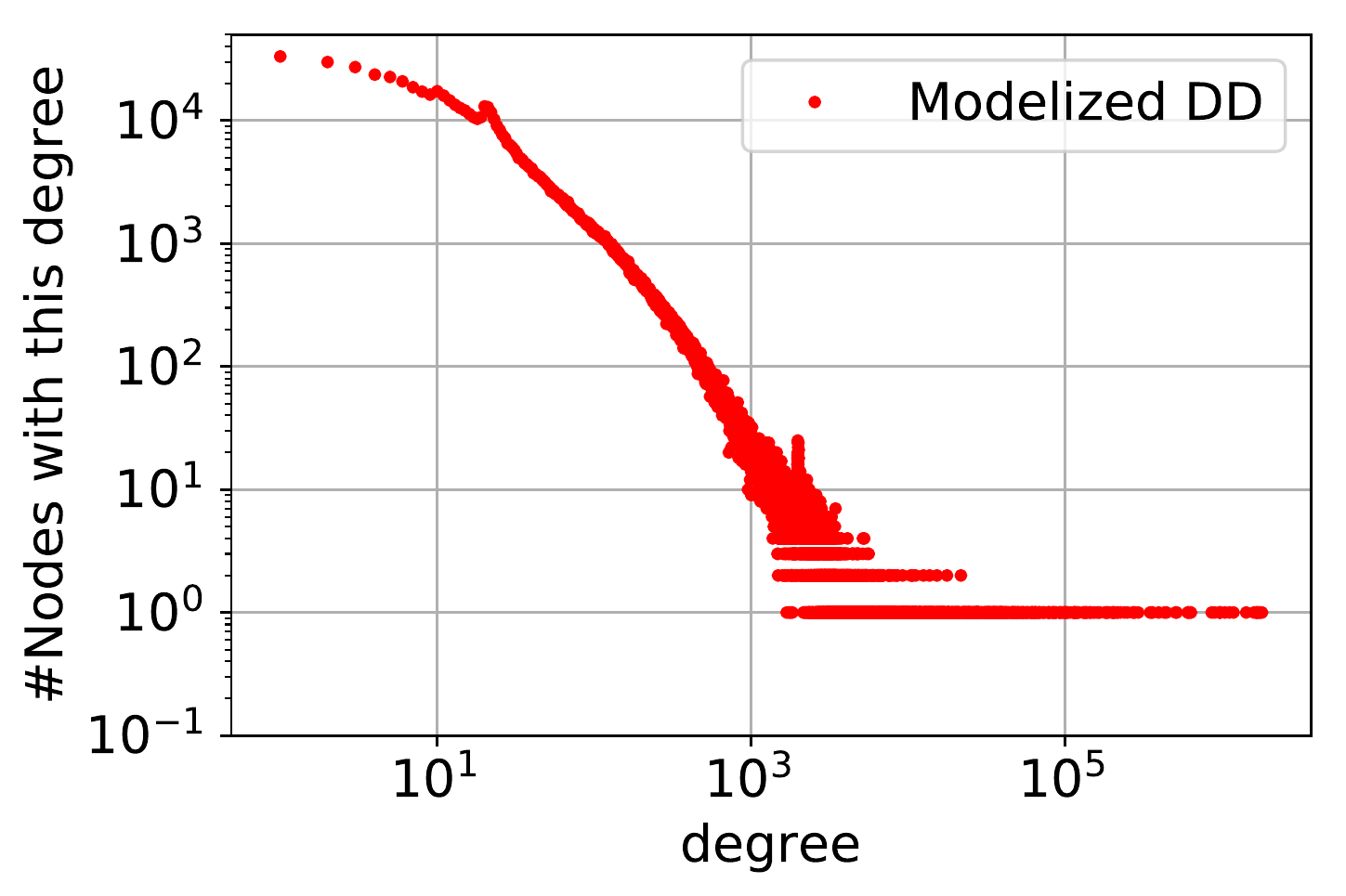}
		\caption{\DD of a random network with $8 \cdot 10^{5}$ nodes using the attachment function of Figure~\ref{fig:Real_f}.}
		\label{fig:Real_DD}
	\end{subfigure}
	\\
	\begin{subfigure}[b]{.2\textwidth}
		\hfill
		\centering
		\includegraphics[scale=0.4]{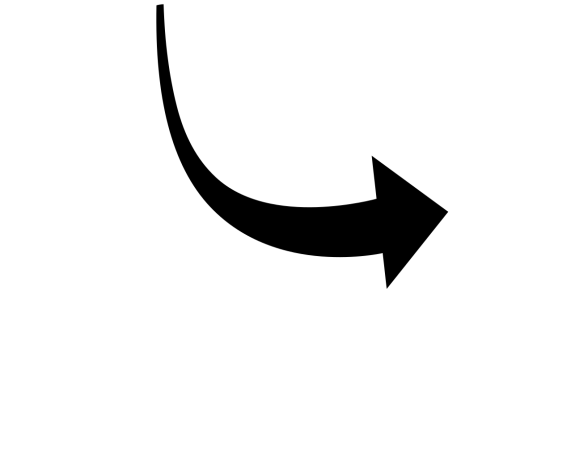}
	\end{subfigure}
	\begin{subfigure}{.5\textwidth}
		\centering
		\includegraphics[scale=0.35]{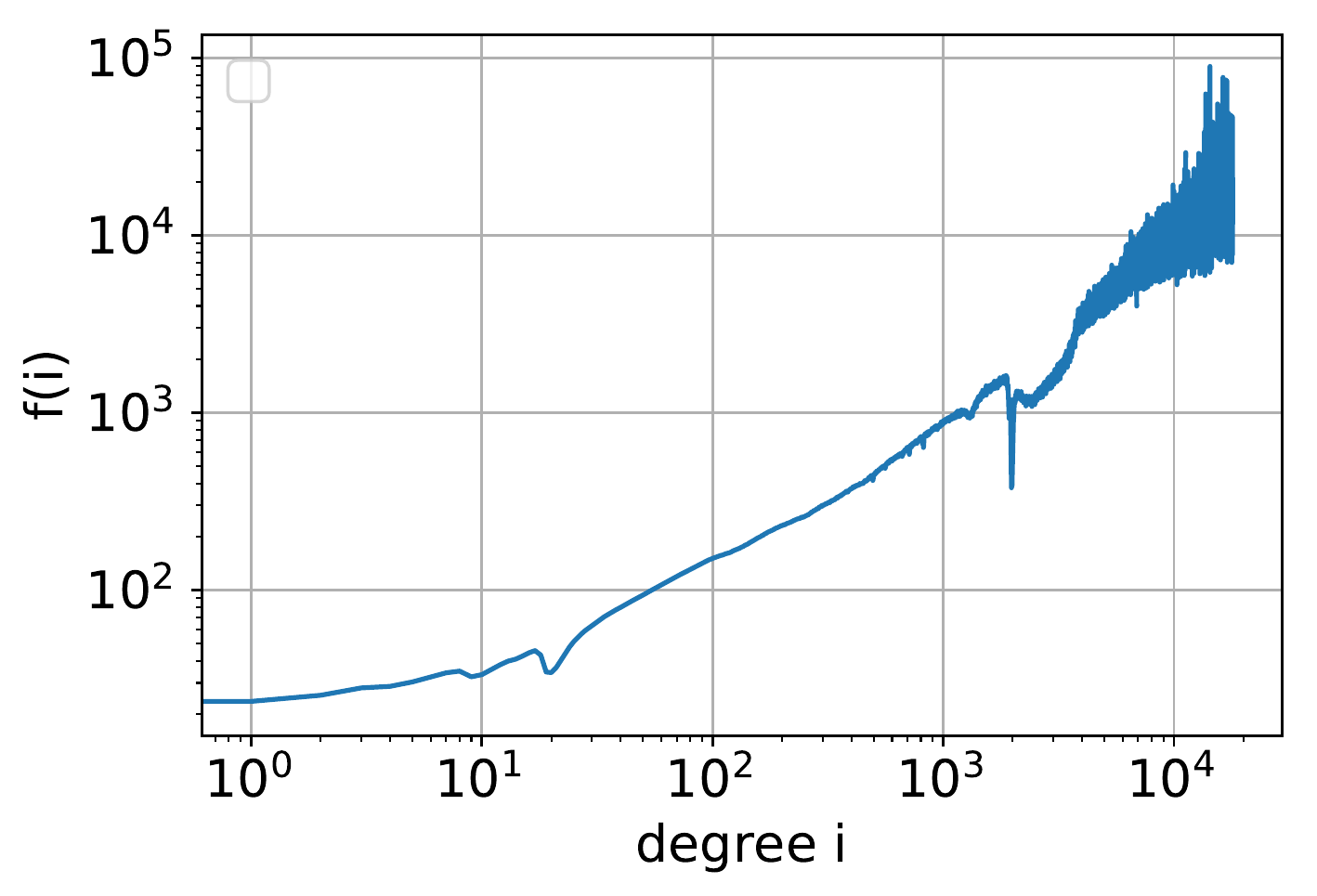}
		\caption{Attachment function $f$ resulting from the undirected \DD of Twitter.}
		\label{fig:Real_f}
	\end{subfigure}
	\begin{subfigure}[b]{.2\textwidth}
		\hfill
		\centering
		\includegraphics[scale=0.4]{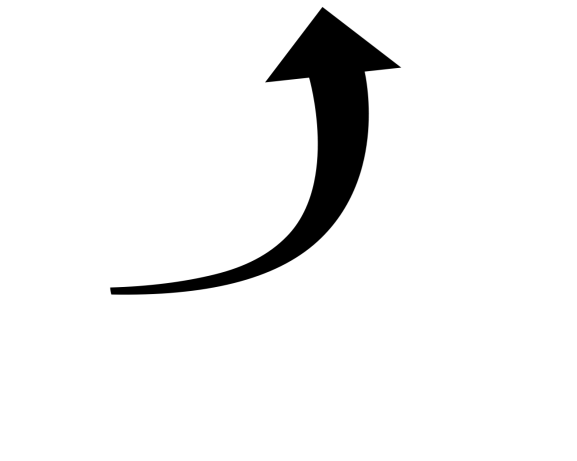}
	\end{subfigure}
	\caption{Modelization of the undirected Twitter's graph.}
	\label{fig:fig}
	\vspace{-0.5cm}
\end{figure}

\vspace{-0.3cm}
\subsection{Undirected \DD of Twitter}
\label{sec:Twitter_DD}
\vspace{-0.2cm}

For this study, we use a Twitter snapshot from 2012, recovered by Gabielkov and Legout~\cite{gabielkov2012complete} and made available by the authors. This network contains 505 million nodes and 23 billion edges, making it one of the biggest social graph available nowadays. Each node corresponds to an account, and an arc $(u,v)$ exists if the account $u$ follows the account $v$. The in- and out-\DDs are presented in~\cite{trolliet2020interest}.

In our case, we look at an undirected version of the Twitter snapshot. We consider the degree of each node as being the sum of its in- and out-degrees. The distribution of this undirected graph is presented in Figure~\ref{fig:Twitter_Undirected_DD}. 
We notice two spikes, around $d=20$ and $d=2000$. We do not know the reason of the first one (which could be social, or due to recommendation system).
The second spike is explained by a specificity of Twitter: until 2015, to avoid bots which were following a very large number of users, Twitter limited the number of possible followings to $\mathrm{max(2000,\text{number of followers})}$. In other words, a user is allowed to follow more than 2000 people only if he is also followed by more than 2000 people. This leads to a lot of accounts with around 2000 followings.
This highlights the fact that some networks have their own specificities, sometimes due to intern policies, which cannot be modeled but by a model specifically built for them.

\vspace{-0.2cm}
\subsection{Modelization}

\noindent
Figure~\ref{fig:Real_f} presents the obtained form of the attachment function $f$ computed using Equation~\ref{eq:thm} with the \DD of Twitter. 
We notice that the overall function is mainly increasing, showing that nodes of higher degrees have a higher chance to connect with new nodes, like in classical preferential attachment models. 
We also notice two drops, around 20 and 2000. They are associated with the risings on the \DD on the same degrees: to increase the amount of nodes with those degrees, the attachment function has to be smaller, so nodes with this degree have less chance to gain new edges. 

We finally use our model with the empirical attachment function of Figure~\ref{fig:Real_f}. Note that, in an empirical study, $P$ can be equal to zero for some degrees, for which no node has this degree in the network. In Twitter, the smallest of those degrees occurs around $18.000$. In that case, $f$ cannot be computed. To get around this difficulty, we interpolate the missing values of $P$, using the two closest smaller and bigger degrees of the missing points. Since we observe the probability distribution on a log-log scale, we interpolate between the two points as a straight line on a log-log scale, i.e., as a power-law function. We believe this is a fair choice since we only look at the tail of the distribution, which looks like a straight line, and since we interpolate between each pair of closest two points only, instead of fitting on the whole tail of the distribution.

The \DD of a random network built with our model is presented in Figure~\ref{fig:Real_DD}. For time computation reasons, the built network only has $N=2 \cdot 10^5$ nodes, to be compared to the $5 \cdot 10^8$ nodes of Twitter. However, it is enough to verify that its DD shape follows the one of the real Twitter's \DD: in particular we recognize the spikes around $d=20$ and $d=2000$.

\begin{longversion}

	\section{Link between the attachment function and heavy-tailed distributions}
	\label{sec:PA_HT}

In this section, we propose to show a correlation between the shape of the attachment function $f$ and the tail of the probability function $P$. 
More precisely, we show that (under some conditions on $f$), if $f$ verifies $lim_{i \rightarrow +\infty} f(i) = + \infty$, then the associated distribution $P$ is heavy-tailed, and if $f$ is bounded from above, then the associated distribution $P$ is not heavy-tailed.

The heavy-tailed feature of \DDs is an interesting property of networks: most of the time, real-world networks exhibit heavy-tailed \DDs, while pure randomness (as we find in the Erdos-Reyni model) build networks with homogeneous \DDs. The particular case of linear preferential attachment is known to build networks with heavy-tailed \DDs. To the best of our knowledge, this is the first time such a general correlation is made between the attachment function of random growing models and the heavy-tailed feature of the \DD. Moreover, if the results presented here only apply for the model proposed in Section~\ref{sec:Model}, we believe the proofs can be extended to almost any other random growing models to show similar results.

Note that we now consider the model in which we impose an attachment function $f$, and we study the shape of the \DD (instead of imposing a probability distribution and study the attachment function, as we have made until now). \\

\subsection{Conditions on $f$}
First of all, $f$ has to verify some conditions in order to give a coherent probability distribution. For instance, choosing $f(i)=i^{\alpha}, \alpha>1$ build a graph where 
a dominant vertex emerges such that after $n$ time steps, the degree of this node is of order n, while the degrees of all other vertices are bounded~\cite{oliveira2005connectivity}. The \DD associated with this attachment function is thus not well-defined. We first express the conditions on $f$. It can be sum up by:

\begin{condition} \label{cond:f-for-HT}
	In order to obtain a distribution $P$ for the \DD verifying $\sum_{k \geq 1} P(k)=1$ and $\sum_{k \geq 1} k P(k)=\mu$, $\mu \in \mathrm{R}^*_+$, 
	the attachment function $f$ has to verify:
	\begin{itemize}
		\item If $f$ converges, $\sum \limits_{i=1}^{+\infty} \frac{(1+\frac{1}{c})^{-i+1}}{f(i)}$ is finite, where $c=\underset{i \geq 1}{max} \Big( f(i) \Big)$;
		\item If $f$ diverges, $\sum \limits_{i=1}^{+\infty} \exp \Big( - \sum \limits_{k=1}^{i} \frac{1}{f(k)} \Big)$ is finite.
	\end{itemize}
\end{condition}
\begin{proof}
	First, we express the condition $\sum_{k \geq 1} k P(k)$ in an interesting form:
	\begin{lemma}\label{lem:fkPk-kPk}
		\begin{equation}\label{eq:fkPk-kPk}
			\sum \limits_{k=1}^{+\infty} f(k)P(k) = \sum \limits_{k=1}^{+\infty} k P(k)
		\end{equation}
	\end{lemma}
	\begin{proof}
		Using Equation~\ref{eq:thm}, we have:
		\begin{align}
			\sum \limits_{k=1}^{+\infty} f(k)P(k) &= \sum \limits_{k=1}^{+\infty} \sum \limits_{k'=k+1}^{+\infty} P(k') \\
			&= \sum \limits_{k=1}^{+\infty} k P(k)
		\end{align}
		\qed
	\end{proof}
	We believe this surprising equality between the two sums might lead to some understandings on links between $P$ and $f$; we keep this exploration for future works.
	We are now left with the study of the convergence of $\sum_{k=1}^{+\infty} P(k)$ and $\sum_{k=1}^{+\infty} f(k)P(k)$.
	\\
	Iterating over Equation~\ref{eq:rec_eq_from_proof} to express $P$ as a function of $f$ gives:
	\begin{equation}
		P(i) = P(1) \prod \limits_{k=2}^i \frac{f(k-1)}{1+f(k)}
	\end{equation}
	We can rewrite this expression as:
	\begin{align}
		P(i) &= P(1) \frac{f(1)}{f(i)} \prod \limits_{k=2}^i \frac{f(k)}{1+f(k)} \\
		&= P(1) \frac{f(1)}{f(i)} \exp \Big( \ln (\prod \limits_{k=2}^i \frac{f(k)}{1+f(k)}) \Big) \\
		&= P(1) \frac{f(1)}{f(i)} \exp \Big( - \sum \limits_{k=2}^i ln( 1 + \frac{1}{f(k)}) \Big) \label{eq:Pi_ln_inter}
	\end{align}
	From now on we distinguish two cases:
	\begin{itemize}
		\item[1)] \underline{$f$ converges}:
		\\
		In this case, $\exists c>0 / \forall i \geq 1, f(i) \leq c$. 
		We have:
		\begin{align}
			P(i) &\leq P(1) \frac{f(1)}{f(i)} \exp \Big( - \sum \limits_{k=2}^i ln(1+\frac{1}{c}) \Big) \\
			&\leq P(1)f(1) \frac{(1+\frac{1}{c})^{-i+1}}{f(i)}
		\end{align}
		So if $f$ converges, $\sum_{k=1}^{+\infty} f(k)P(k)$ always converges, and, by by Lemma~\ref{lem:fkPk-kPk}, the mean of $P$ is finite. The condition on $\sum_{k=1}^{+\infty} P(k)$ gives that $\sum_{k \geq 1} \frac{(1+\frac{1}{c})^{-k}}{f(k)}$ has to be finite.
		\\
		\item[2)] \underline{$f$ diverges}:
		\\
		Then, we can find $i_0$ such that $\sum \limits_{k=2}^i ln( 1 + \frac{1}{f(k)}) \underset{i \rightarrow +\infty}{\sim} \sum \limits_{k=2}^{i_0} ln( 1 + \frac{1}{f(k)}) + \sum \limits_{k=i_0}^{i} \frac{1}{f(k)}$.
		We can rewrite Equation~\ref{eq:Pi_ln_inter} as:
		\begin{align}
			P(i) &\sim P(1) \frac{f(1)}{f(i)} \exp \Big( - \sum \limits_{k=2}^{i_0} ln( 1 + \frac{1}{f(k)}) + \sum \limits_{k=1}^{i_0-1} \frac{1}{f(k)} - \sum \limits_{k=1}^{i} \frac{1}{f(k)} \Big) \\
			&\sim K_{f,i_0} \frac{1}{f(i)} \exp \Big( - \sum \limits_{k=1}^{i} \frac{1}{f(k)} \Big)
		\end{align}
		with $K_{f,i_0}$ a constant depending of $f$ and $i_0$. Thus by Lemma~\ref{lem:fkPk-kPk}, the mean of $P$ is finite if and only if the following quantity is finite: 
		\[
		\sum \limits_{i=1}^{+\infty} \exp \Big( - \sum \limits_{k=1}^{i} \frac{1}{f(k)} \Big).
		\]
		Note that the other condition, i.e. the convergence of $\sum \limits_{i=1}^{+\infty} \frac{1}{f(i)} \exp \Big( - \sum \limits_{k=1}^{i} \frac{1}{f(k)} \Big)$, is included in the first one: indeed, since $f$ diverges, there exists a constant $i_0$ such that $\forall i \geq i_0, \frac{1}{f(i)} \leq 1$, and the second condition can be bounded by the first one.
	\end{itemize}
	\qed
\end{proof}
It is interesting to note that, for $f(i)\propto i^{\alpha}$, $\alpha=1$ is the limit case where Condition~\ref{cond:f-for-HT} holds, as expected from the results of~\cite{oliveira2005connectivity}.

\subsection{Link between the limit of $f$ and heavy-tailed \DDs}

\begin{definition}~\cite{rolski2009stochastic}
	We say that a distribution $P$ is \textit{heavy-tailed} if it decays more slowly than an exponential, i.e.: 
	\[
	\forall t>0, e^{ti} P(X>i) \underset{i \rightarrow +\infty}{\rightarrow} +\infty.
	\]
\end{definition}
We show the two following theorems:
\begin{theorem}
	Let f be an attachment function verifying Condition~\ref{cond:f-for-HT} and such that $\underset{i \rightarrow +\infty}{lim} f(i) = + \infty$. Then the associated distribution $P$ is heavy-tailed.
	\label{thm:PA_HT-1}
\end{theorem}
\begin{theorem}
	Let f be an attachment function verifying Condition~\ref{cond:f-for-HT} and such that f is bounded from above by $M>0$. Then the associated distribution $P$ is not heavy-tailed.
	\label{thm:PA_HT-2}
\end{theorem}





\noindent
To prove those theorems, we will use the following lemma:

\begin{lemma}
	$P$ is heavy-tailed if and only if 
	\[
	\forall t>0, \exists i_0 > 0 / \underset{i \rightarrow + \infty}{lim}  g_{t,i_0}(i)=+\infty,
	\]
	where $g_{t,i_0}(i) = ti + log(f(i_0)) - \sum \limits_{k=i_0}^{i-1} log(1 + \frac{1}{f(k+1)})$.
\end{lemma}

\begin{proof}
	We recall that $P(i)=P(1) \prod \limits_{k=1}^{i-1} \frac{f(k)}{1+f(k+1)}$ and $f(i)=\frac{1}{P(i)} \sum \limits_{k=i+1}^{\infty} P(k)$. It implies \[P(X>i) := \sum \limits_{k=i+1}^{\infty} P(k) = f(i)P(i) = f(i) P(1) \prod \limits_{k=1}^{i-1} \frac{f(k)}{1+f(k+1)}\].
	\\
	Let $t>0$, $i_0>0$. We have:
	\begin{align*}
		e^{ti} P(X>i) &= e^{ti} f(i) P(1) \prod \limits_{k=1}^{i-1} \frac{f(k)}{1+f(k+1)} \\
		&= e^{ti} e^{log(f(i))} P(1) \prod \limits_{k=1}^{i_0-1} \frac{f(k)}{1+f(k+1)} \prod \limits_{k=i_0}^{i- 1} e^{log(\frac{f(k)}{1+f(k+1)})} \\
		&= P(1) \prod \limits_{k=1}^{i_0-1} \big( \frac{f(k)}{1+f(k+1)} \big) \times e^{ti + log(f(i)) + \sum \limits_{k=i_0}^{i-1} log \big( \frac{f(k)}{1+f(k+1)} \big)}
	\end{align*}
	We call $g_{t,i_0}(i) = ti + log(f(i)) + \sum \limits_{k=i_0}^{i-1} log (\frac{f(k)}{1+f(k+1)})$. 
	$P$ is heavy-tailed if and only if $\underset{i \rightarrow + \infty}{lim}  g_{t,i_0}(i)=+\infty$.
	But $g_{t,i_0}$ can also be expressed as:
	
	\begin{align*}
		g_{t,i_0}(i) &= ti + log(f(i)) - \sum \limits_{k=i_0}^{i-1} log (\frac{1+f(k+1)}{f(k)})
		\\
		&= ti + log(f(i)) - \sum \limits_{k=i_0}^{i-1} log (\frac{f(k+1)}{f(k)} (1 + \frac{1}{f(k+1)}))
		\\
		&= ti + log(f(i)) - \sum \limits_{k=i_0}^{i-1} log (f(k+1)) + \sum \limits_{k=i_0}^{i-1} log (f(k)) - \sum \limits_{k=i_0}^{i-1} log(1 + \frac{1}{f(k+1)})
		\\
		&= ti + log(f(i)) - log(f(i)) + log(f(i_0)) - \sum \limits_{k=i_0}^{i-1} log(1 + \frac{1}{f(k+1)})
		\\
		&= ti + log(f(i_0)) - \sum \limits_{k=i_0}^{i-1} log(1 + \frac{1}{f(k+1)}).
	\end{align*}
	\qed
\end{proof}


\noindent
\textbf{Proof of Theorem~\ref{thm:PA_HT-1}. \\}
Let $t>0$.
By definition of the limit, $\exists i_0 / \forall i > i_0, f(i) > \frac{1}{e^{t/2}-1}$. So:

\begin{align}
	g_{t,i_0}(i) &= ti + log(f(i_0)) - \sum \limits_{k=i_0}^{i-1} log(1 + \frac{1}{f(k+1)})
	\\
	&> ti + log(f(i_0)) - \sum \limits_{k=i_0}^{i-1} log(1 + \frac{1}{(\frac{1}{e^{t/2}-1})})
	\\
	&= ti + log(f(i_0)) - (i-i_0-1) \frac{t}{2}
	\\
	&= \frac{1}{2} i + log(f(i_0)) + (i_0+1) \frac{t}{2}
	\\
	&\underset{i \rightarrow +\infty}{\rightarrow} +\infty
\end{align}
\qed


\noindent
\textbf{Proof of Theorem~\ref{thm:PA_HT-2}. \\}
\begin{align}
	g_{t,i_0}(i) &= ti + log(f(i_0)) - \sum \limits_{k=i_0}^{i-1} log(1 + \frac{1}{f(k+1)})
	\\
	&< ti + log(f(i_0)) - \sum \limits_{k=i_0}^{i-1} log(1 + \frac{1}{M})
	\\
	&= ti + log(f(i_0)) - (i-i_0-1) log(1 + \frac{1}{M})
\end{align}
Let $t = \frac{1}{2} log(1 + \frac{1}{M})$.
\begin{align}
	g_{t,i_0}(i) &= - \frac{1}{2} log(1 + \frac{1}{M}) i + log(f(i_0)) + (i_0+1) log(1 + \frac{1}{M})
	\\
	&\underset{i \rightarrow +\infty}{\rightarrow} -\infty.
\end{align}
There exists a value of $t>0$ such that the limit of $g_{t,i_0}$ goes to $-\infty$, hence $P$ is not heavy-tailed.
\qed	

\begin{remark}
	Preferential attachment functions (i.e., increasing functions) set is not included nor it contains any of previous cases: we can have a preferential attachment function in the first case, as well as in the second case; we can have an non preferential attachment function in the first case, as well as in second case.
\end{remark} 

\begin{remark}
	Not all functions are included in the previous cases: it remains the cases where the limit of $f$ is not infinite but f is not bounded either (for instance, $f(i)=1$ if $i$ is pair, $f(i)=i$ otherwise). However, we state those cases are quite rare.
\end{remark}



\end{longversion}

\section{Conclusion}

In this paper, we proposed a new random growth model picking the nodes to be connected together in the graph with a flexible probability $f$. We expressed this $f$ as a function of any distribution $P$, leading to the possibility to build a random network with any wanted degree distribution. We computed $f$ for some classical distributions, as much as for a snapshot of Twitter of 505 million nodes and 23 billion edges. We believe this model is useful for anyone studying networks with atypical degree distributions, regardless of the domain. If the presented model is undirected, we also believe a directed version of it, based on the Bollobás et al. model~\cite{bollobas2003directed}, can be easily generalized from the presented one. 
We also hope it can enlighten relations between the degree distributions of networks and the attachment function behind them, both in random growth models as well as real-world networks. To take a step in that direction, we show that, in our model, the limit of the attachment function $f$ is sufficient to determine if the probability distribution of the graphs is heavy-tailed or not. We believe this result can be extended to other models, and hopefully lead to interesting studies on real-world networks.


\bibliographystyle{plain}
\bibliography{biblio.bib}

\end{document}